\documentclass[a4paper,conference]{IEEEtran}

\usepackage{indentfirst}
\usepackage[T1]{fontenc}
\usepackage[]{graphicx}
\usepackage[]{epsfig}
\usepackage{amsmath}
\usepackage{amssymb}
\usepackage{verbatim}
\usepackage{tikz} 
\usepackage{bbm}
\usepackage{enumerate}
\usepackage{pdfsync}
\usepackage{algorithm}
\usepackage{algpseudocode}
\usepackage{array}
\usepackage{mathrsfs}
\usepackage{dsfont}
\usepackage[linktoc=page,colorlinks=true,linkcolor=titres,citecolor=titres]{hyperref}
\usepackage{bookmark}
\usepackage{makeidx}
\usepackage{amsfonts}
\usepackage{amsthm}
\usepackage{url}
\usepackage[caption=false,font=footnotesize]{subfig}
\usepackage{cite}
\usepackage{color}
\usepackage{tocvsec2}
\usepackage{tabularx}

\newtheorem{lemma}{Lemma}

\newtheorem{theorem}{Theorem}
\newtheorem{definition}{Definition}

\newcommand{\E}[1]{\mathbf{E}_{}\!\left[#1\right]}

\def\P{{\mathbf P}}
\def \d{\mathrm{d}}

\usetikzlibrary{arrows} 

\definecolor{clair}{rgb}{0.910,0.933,0.957}
\definecolor{moyen}{rgb}{0.102,0.349,0.557}
\definecolor{fonce}{rgb}{0.067,0.231,0.369}
\definecolor{titres}{rgb}{0.137,0.466,0.741}
\definecolor{apricot}{rgb}{0.961,0.506,0.216}
\definecolor{grey}{rgb}{0.58,0.58,0.58}

\begin{document}

\pgfdeclarelayer{background layer} 
\pgfdeclarelayer{foreground layer} 
\pgfsetlayers{background layer,main,foreground layer}

\title{Computing the $k$-coverage of a wireless network}
\author{\IEEEauthorblockN{Ana\"is Vergne,
Laurent Decreusefond, and
Philippe Martins}\\
\IEEEauthorblockA{LTCI, T\'el\'ecom ParisTech, Universit\'e Paris-Saclay, 75013, Paris, France}
}

\maketitle
\begin{abstract}
Coverage is one of the main quality of service of a wireless
network. $k$-coverage, that is to be covered simultaneously by $k$
network nodes,  is synonym of reliability and numerous applications
such as multiple site MIMO features, or handovers. We introduce here a
new algorithm for computing the $k$-coverage of a wireless
network. Our method is based on the observation that $k$-coverage can
be interpreted as $k$ layers of $1$-coverage, or simply coverage. We
use simplicial homology to compute the network's topology and a
reduction algorithm to indentify the layers of $1$-coverage. We
provide figures and simulation results to illustrate our algorithm.
\end{abstract}

\section{Introduction}
Wireless networks encompass cellular networks, WiFi access points,
sensor networks, and so on. With the increasing usage of high data
rates mobile devices such as smartphones and tablets, and the
development of the Internet of Things (IoT), they have become
indispensable in our everyday lives. A common key quality of service
of this type of networks is the coverage. The coverage of a wireless
network is the set of points that are in the sensing range of at least
one network node. The greater the covered area, the more mobile
devices can have access to it. For cellular networks, coverage can
even be a governmental obligation. The absence of coverage holes inside
the covered area is needed to offer a continuous access to services.

However, network nodes are often not regularly deployed on lattice or
according to the hexagonal model in practice, see
\cite{gomez_case_2015} for cellular networks in France for example.
And, deciding whether a set of network nodes does cover
a given area is not that easy for arbitrary deployments. Simplicial
homology can help us do that, considering the network nodes GPS
positions and their coverage ranges, it is possible to build a purely
combinatorial object, namely an abstract simplicial complex, of which
it is possible to compute the topology. Basically an abstract simplicial
complex is the generalization of the concept of graph, it is made of
$k$-simplices where $0$-simplices are vertices, $1$-simplices are
edges, $2$-simplices are triangles, $3$-simplices are tetrahedron and
so on. In particular, geometrical simplicial complexes such as the
\u{C}ech complex and the Vietoris-Rips complex, represent exactly and
approximatively respectively, the topology of the union of the
coverage disks as stated in~\cite{ghrist_coverage_2005}. Then
algebraic topology, \cite{hatcher_algebraic_2002}, is a mathematical
tool that can compute the number of connected components, of coverage
holes, and of 3D voids, namely the Betti numbers of the
simplicial complex representing the network, as explained in
\cite{de_silva_coordinate-free_2006}. Since, the computational time
to obtain the Betti numbers can explode with the size of the
simplicial complex, many works focus on faster ways to compute them,
for instance in a decentralized way \cite{muhammad_decentralized_2007}, 
using persistent homology \cite{zomorodian_computing_2005}, thanks to
chain complexes reduction \cite{kaczynski_homology_1998} , or with
witness complexes reduction \cite{de_silva_topological_2004}. In our
work we use simplicial complex reduction to reduce a simplicial
complex to the minimum number of points needed to provide
$1$-coverage. Precisely, we use the reduction algorithm presented in 
\cite{vergne_reduction_2013}, that can also be found for coverage hole
detection in \cite{yan_homology-based_2015} and for energy efficiency
in cellular networs in \cite{vergne_simplicial_2015}.

Coverage can thus be computed mathematically thanks to algebraic
topology. However $k$-coverage computing is not that simple. Indeed, a
point is said to be $k$-covered when it is in the covered area of at
least $k$ network nodes. Consequently, an area is $k$-covered whenever
every point in it is $k$-covered. The expected $k$-coverage in
wireless networks has been studied in \cite{yen_expected_2006} in
order to propose a node scheduling scheme that conserves energy while
retaining network coverage. In \cite{li_ensuring_2017}, the authors
use $k$-order Vorono\"i diagrams to compute the density of network
nodes required to achieve $k$-coverage.

In this article, we propose a method and an algorithm for computing
the $k$-coverage of a given wireless network. Theoretically it is easy
to compute the probability for a point to be $k$-covered for a
wireless network generated by a Poisson point process. However it is
more difficult to apprehend the $k$-coverage of a whole
area. Moreover, probabilistic results can not be applied to every
wireless network. That is why we need simplicial homology
representation to compute the topology of a given network as a whole. We
then exhibit that the $k$-coverage can be seen as $k$ layers of
$1$-coverage and give an algorithm that compute the $k$-coverage of a
wireless network. For operating purposes, the $k$ layers of network
nodes that ensure each $1$-coverage are returned by our algorithm.

First in Section \ref{sec_kcov}, we define the $k$-coverage and
discuss its application for wireless networks such as IoT sensor
networks and cellular networks. We present a probabilistic network model and give
some theoretical results in Section \ref{sec_mod}, and introduce
few needed mathetical tools in Section \ref{sec_simp}. Then in Section
\ref{sec_alg}, we give our algorithm for computing $k$-coverage, and
simulation results in Section \ref{sec_sim}. Finally we conclude in
Section \ref{sec_ccl}.

\section{$k$-coverage}
\label{sec_kcov}
Cellular networks, Wireless Local Area Networks (WLANs), and sensor
networks take part in the family of wireless networks. In these
networks, coverage define the utility of the network. In cellular
networks or WLANs, users can access the service only if they are
covered by a network node. In sensor networks, sensor can communicate
only if they are in the sensing range of each other.
In a wireless network, a point is said to be covered if it is in the
sensing range of a network node, that is to say if it is in the
coverage of this node. An area is then covered, when every point of it
is covered. By extension, an area is $k$-covered when every point of
the area is in the coverage of at least $k$ network nodes. A wireless
network providing $k$-coverage for an area with a large $k$ is then a
densely deployed network.

In the literature, $k$-coverage is more often used for sensor networks such as
Low-Power Wide Area Networks (LPWANs), and Internet of Things (IoT)
\cite{yen_expected_2006, li_ensuring_2017}.
Indeed the benefits of $k$-coverage include better reliability, better
accuracy in sensor measurements, greater throughput by using multiple
channels, etc. And these uses concern primarily sensor nodes.
Moreover, sensor nodes are small, live on battery, and are cheap to
buy and replace, so they can be deployed in large quantities, thus
providing easily $k$-coverage with a great $k$.

However, $k$-coverage can also be of interest for cellular networks, where
it becomes synonym of multi-site transmitter. The first application is
the handover. Actually, a handover is performed when a user changes cells during a
communication. It is called a soft-handover when the user is connected
simultaneously to multiple cells for the transition between cells, or
for interference mitigation in dense area in 3G networks.
Therefore, a user in a $k$-covered area would have the
possibility to have handovers with $k$ different cells, which means
that telecommunication operators could make trafic off-loading
decisions by directing users to less-busy cells, or offer better radio
channels thanks to antenna diversity. In 4G and later networks,
$k$-coverage means also that MIMO transmissions from multiple base
stations to a user can be performed. That is the basis of the
Coordinated Multi-Point Joint Processing (CoMP JP) scheme that allows
great capacity gains \cite{vergne_evaluating_2010}. In CoMP JP, two or
more base stations can cooperate to serve simultaneously a user leading to a
throughput multiplied by $2$ or more.

\section{Probabilistic analysis}
\label{sec_mod}
We represent the wireless network nodes by a Poisson point process:
\begin{definition}
Let $\d \mu=\lambda \d x$ be the Lebesgue measure on $E \subseteq \mathbb{R}^2$, $N$ is a
spatial Poisson point process of intensity $\lambda >0$ on $E$ if:
\begin{itemize}
\item $N(A)$ the number of points that fall in $A \subset E$ follows a
  Poisson law
$$ \P[N(A)=k] = e^{-\mu(A)}\frac{\mu(A)^k}{k!}$$
\item If $A, B \subset E$ such that $A \cap B = \emptyset$, then
  $N(A)$ and $N(B)$ are independant.
\end{itemize}
\end{definition}
We can note that $\mu(A)=\lambda S(A)$ where $S(A)$ is the area of
$A$. Moreover, conditionnally to $N(A)=n$ for $A \subset E$, then the
points of $N$ are independantly and uniformly distributed on $A$.

We suppose that every network node has the same sensing range $r$. Thus
its coverage area is a disk of radius $r>0$. Then for $x \in E$, the
probability for $x$ to be $k$-covered is given by:
\begin{eqnarray*}
\P[x \text{ } k\text{-covered}] = \P[\exists y_1,\dots,y_k \in N |
                                      x \in \bigcap_{i=1}^kB(y_i,r)],
\end{eqnarray*}
where $B(y,r)$ is the ball of center $y \in E$ and radius $r>0$.
We immediatly have that:
\begin{eqnarray*}
\P[x \text{ } k\text{-covered}]\!\! &=& \!\!\P[\exists y_1,\dots,y_k \in N |
                                      x \in \bigcap_{i=1}^kB(y_i,r)]\\
&=&\!\!\P[\exists y_1,\dots,y_k \in N | \forall i, y_i \in B(x,r)]\\
&=&\!\!\P[N(B(x,r))\geq k]\\
&=&\!\! 1 - \P[N(B(x,r)) < k]\\
&=&\!\! 1-\sum_{i=0}^{k-1}e^{-\lambda \pi r^2}\frac{(\lambda \pi r^2)^i}{i!}.
\end{eqnarray*}

We can see in Fig.~\ref{fig_courbes} the probability for a point to be
$k$-covered for $k=1,\dots,6$, depending on $\lambda \pi r^2$. As $r$
is fixed, only $\lambda$ varies. Logically, as $\lambda$ grows, the
probability to be $k$-covered tends to $1$, and the greater $k$ is,
the smaller the probability to be $k$-covered is.
\begin{figure}[h]
  \centering
    \includegraphics[width=7cm]{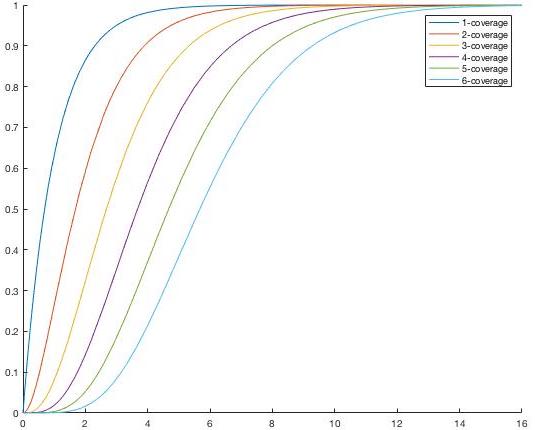}
  \caption{Probability for a point to be $k$-covered depending on
    $\lambda \pi r^2$}
\label{fig_courbes}
\end{figure}

The probability that a point $x \in E$ is exactly $k$-covered and not
$(k+1)$-covered is then:
\begin{eqnarray*}
\P[x \text{ exactly } k\text{-covered}] &=& \P[N(B(x,r))=k]\\
&=&e^{-\lambda \pi r^2}\frac{(\lambda \pi r^2)^k}{k!}.
\end{eqnarray*}

We can derive the mean $k$ for which a point $x \in E$ is $k$-covered
and not $(k+1)$-covered:
\begin{eqnarray*}
\E{k}&=&\sum_{k=1}^\infty k \P[x \text{ exactly } k\text{-covered}]\\
&=&\sum_{k=1}^\infty e^{-\lambda \pi r^2}\frac{(\lambda \pi
    r^2)^k}{(k-1)!}\\
&=&e^{-\lambda \pi r^2} \sum_{k=0}^\infty \frac{(\lambda \pi
    r^2)^k}{k!} (\lambda \pi r^2) \\
&=& \lambda \pi r^2.
\end{eqnarray*}
Therefore, the mean $k$ for which a point $x \in E$ is $k$-covered and
not $(k-1)$-covered is directly proportional to $\lambda$.


However, considering the probability for a point to be $k$-covered is
not sufficient to benefit from the advantages of $k$-coverage. Indeed,
in wireless networks, reception devices such as phones are
mobile. Therefore, one needs a whole area to be $k$-covered to offer
$k$-coverage applications such as joint processing or handover.
The probability of $k$-coverage of one point is an upper-bound of the
probability of a whole area to be completely $k$-covered without any
hole. Indeed it is easier to ensure that a point is covered, than a
whole area without any hole. So probabilistic results can not be used
for engineering purposes. That is why we need to consider another
approach : to study and compute the coverage of the network as a
whole, that is mathematically to study the topology of the network.

\section{Simplicial homology and algebraic topology}
\label{sec_simp}
Considering a set of points representing network nodes, the first idea
to apprehend the topology of the network would be to look at the
neighbors graph: if the distance between two points is less than a
given parameter then an edge is drawn between them. However this
representation is too limited to transpose the network's
topology. First, only $2$-by-$2$ relationships are represented in the
graph, there is no way to grasp interactions between three or more
nodes. Moreover, there is no concept of coverage in a graph. That is
why we are interested in more complex objects.

Indeed, graphs can be generalized to more generic combinatorial
objects known as simplicial complexes. While graphs model binary
relations, simplicial complexes can represent higher order
relations. A simplicial complex is thus a combinatorial object made up
of vertices, edges, triangles, tetrahedra, and their $n$-dimensional
counterparts. 
Given a set of vertices $X$ and an integer $k$, a $k$-simplex is an
unordered subset of $k+1$ vertices $\{x_0,\dots, x_k\}$ where $x_i\in
X, \forall i \in \{0,\dots,k\}$ and $x_i\not=x_j$ for all
$i\not=j$. Thus, a $0$-simplex is a vertex, a $1$-simplex an edge, a
$2$-simplex a triangle, a $3$-simplex a tetrahedron, etc. See
Fig.~\ref{fig_simplices} for instance. 

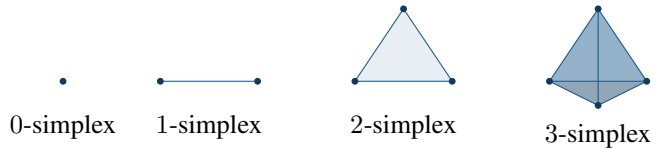
\begin{figure}[h]
  \centering
    \begin{tikzpicture}[scale=0.64]
\fill [color=fonce] (0,0) circle (2pt);
\node [below] at (0,-0.5) {$0$-simplex};
\draw [color=moyen] (2,0)--(4,0);
\fill [color=fonce] (2,0) circle (2pt);
\fill [color=fonce] (4,0) circle (2pt);
\node [below] at (3,-0.5) {$1$-simplex};
\fill [color=clair] (6,0)--(8,0)--(7,1.5);
\draw [color=moyen] (6,0)--(8,0)--(7,1.5)--(6,0);
\fill [color=fonce] (6,0) circle (2pt);
\fill [color=fonce] (8,0) circle (2pt);
\fill [color=fonce] (7,1.5) circle (2pt);
\node [below] at (7,-0.5) {$2$-simplex};
\fill [color=moyen, opacity=1] (10,0)--(12,0)--(11,1.5);
\fill [color=fonce,opacity=1] (10,0)--(12,0)--(11,-0.5);
\fill [color=clair,opacity=0.6] (10,0)--(11,-0.5)--(11,1.5);
\fill [color=clair,opacity=0.6] (11,-0.5)--(12,0)--(11,1.5);
\draw [color=moyen] (10,0)--(12,0)--(11,1.5)--(10,0);
\draw [color=moyen] (10,0)--(11,-0.5)--(11,1.5);
\draw [color=moyen] (12,0)--(11,-0.5);
\fill [color=fonce] (10,0) circle (2pt);
\fill [color=fonce] (12,0) circle (2pt);
\fill [color=fonce] (11,1.5) circle (2pt);
\fill [color=fonce] (11,-0.5) circle (2pt);
\node [below] at (11,-0.7) {$3$-simplex};
    \end{tikzpicture}
  \caption{Examples of $k$-simplices.}\label{fig_simplices}
\end{figure}

Any subset of vertices included in the set of the $k+1$ vertices of a
$k$-simplex is a face of this $k$-simplex. A $k$-face is then a face
that is a $k$-simplex. The inverse notion of face is coface. An
abstract simplicial complex is a set of simplices 
such that all faces of these simplices are also in the set of
simplices.

In this article, we are intersted in representing the topology of a
wireless network, we introduce the two following abstract simplicial
complexes: 
\begin{definition}[\u{C}ech complex]
 Let $\omega$ be a finite set of points in
  $\mathbb{R}^2$, and $r$ a real positive number. The \u{C}ech
  complex of parameter $r$ of $\omega$, 
  $\mathcal{C}_{r}(\omega)$, is the abstract simplicial complex
  whose $k$-simplices correspond to the unordered $(k+1)$-tuples of
  vertices in $\omega$ such that the intersection of the $k+1$ balls
  centered on them is non empty.
\end{definition}

\begin{definition}[Vietoris-Rips complex]
Let $\omega$ be a finite set of points in
  $\mathbb{R}^2$, and $\epsilon$ a real positive number. The Vietoris-Rips
  complex of parameter $\epsilon$ of $\omega$, 
  $\mathcal{R}_{\epsilon}(\omega)$, is the abstract simplicial complex
  whose $k$-simplices correspond to the unordered $(k+1)$-tuples of
  vertices in $\omega$ which are pairwise within distance less than
  $\epsilon$ of each other.
\end{definition}

The \u{C}ech complex provides the representation of the exact topology
of the network (see the Nerve lemma in
\cite{ghrist_coverage_2005}) but can be tricky to compute due to the
check of whether three disks intersect or not.
One can see easily that the Vietoris-Rips complex
$\mathcal{R}_{2r}(\omega)$ is an approximation of the \u{C}ech complex
$\mathcal{C}_{r}(\omega)$ that is way easier to compute since it is a
clique complex based only on the neighbors graph information.
This approximation is quite good:
in the case of a random uncorrelated deployment with
network nodes deployed according to a Poisson point process the error
is less than $0.06\%$ in the computation of the covered area
\cite{yan_accuracy_2012}. 
An example of a \u{C}ech complex
representing a wireless network can be seen in Fig.~\ref{fig_net}. We
can see $4$ coverage holes in the network that are highlighted in the
simplicial complex representation.
\begin{figure}[h]
  \centering
    \includegraphics[width=4.35cm]{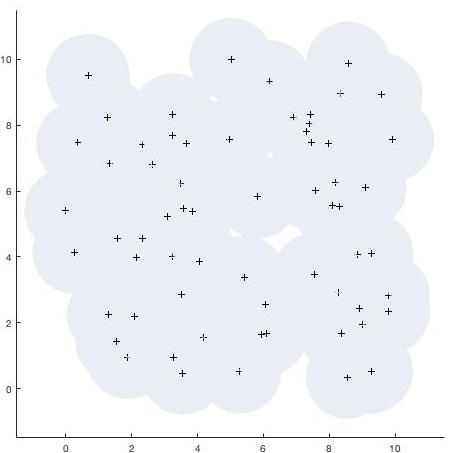}
 \includegraphics[width=4.35cm]{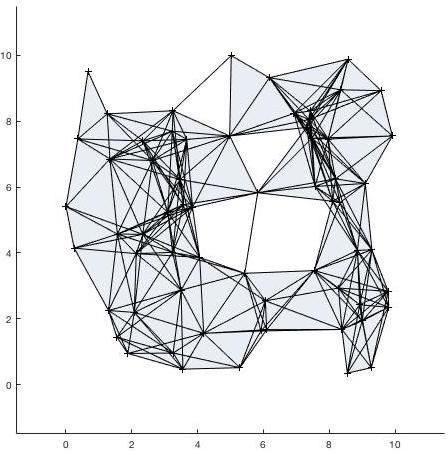}
  \caption{ A \u{C}ech complex representing a wireless network}
\label{fig_net}
\end{figure}

Given an abstract simplicial complex, one can define an orientation on
the simplices by defining an order on the vertices, where a change in the
orientation, that is a swap between two vertices, corresponds to a
change in the sign. 
Then let us define the vector spaces of the $k$-simplices of a
simplicial complex, and the associated boundary maps:
\begin{definition}
  Let $S$ be an abstract simplicial complex.

 For any integer $k$,
  $\mathscr{C}_k(S)$ is the vector space spanned by the set of oriented
  $k$-simplices of $S$.
\end{definition}

\begin{definition}
Let $S$ be an abstract simplicial complex and $\mathscr{C}_k(S)$ the
vector space of its $k$-simplices for any $k$ integer.

The boundary map $\partial_k$ is defined as the linear
  transformation $\partial_k: \mathscr{C}_k(S)\rightarrow
  \mathscr{C}_{k-1}(S)$ which acts on the 
  basis elements $[x_0,\dots,x_k]$ of $\mathscr{C}_k(S)$ via:
  \begin{eqnarray*}
    \partial_k [x_0,\dots,x_k]= \sum_{i=0}^{k}{(-1)}^i [x_0,\dots,x_{i-1},x_{i+1},\dots,x_k].
  \end{eqnarray*}
\end{definition}

For example, for a $2$-simplex we have:
\begin{figure}[H]
\centering 
\begin{tikzpicture}[scale=0.7]
\coordinate (x1) at (0,0);
\coordinate (x2) at (1,1);
\coordinate (x3) at (2,0);
\coordinate (x4) at (5,0);
\coordinate (x5) at (6,1);
\coordinate (x6) at (7,0);
 \fill [color=clair] (x1)--(x2)--(x3);
\draw[draw=moyen] (x1)--(x2);
\draw[draw=moyen](x2)--(x3);
\draw[draw=moyen](x3)--(x1);
\draw[draw=moyen, arrows={-triangle 45}] (x4)->(x5);
\draw[draw=moyen, arrows={-triangle 45}] (x5)->(x6);
\draw[draw=moyen, arrows={-triangle 45}] (x6)->(x4);
\draw[->][color=fonce] (1,.2) arc (270:-30:.2);
\draw [fill=fonce] (x1) circle (2pt);
\draw [fill=fonce] (x2) circle (2pt);
\draw [fill=fonce] (x3) circle (2pt);
\draw [fill=fonce] (x4) circle (2pt);
\draw [fill=fonce] (x5) circle (2pt);
\draw [fill=fonce] (x6) circle (2pt);
\node [below] at (x1) {$x_0$};
\node [above] at (x2) {$x_1$};
\node [below] at (x3) {$x_2$};
\node [below] at (x4) {$x_0$};
\node [above] at (x5) {$x_1$};
\node [below] at (x6) {$x_2$};
\node [below] at (1,-0.5) {\small $\partial_2([x_0,x_1,x_2])$};
\node [below] at (3,-0.75) {\small $=$};
\node [below] at (6,-0.5) {\small $[x_1,x_2]-[x_0,x_2]+[x_0,x_1]$};
 \end{tikzpicture}
\end{figure}

As its name indicates, the boundary map applied to a linear
combination of simplices gives its boundary. The boundary of a
boundary is the null application. Therefore the following theorem can be easily
demonstrated (see~\cite{hatcher_algebraic_2002} for instance):
\begin{theorem}
  For any $k$ integer, 
$\partial_k \circ\partial_{k+1}=0.$
\end{theorem}

Let $S$ be an abstract simplicial complex. Then we can denote 
the $k$-th boundary group of $S$ as $B_k(S)=\mathrm{im}
  \, \partial_{k+1}$, and the $k$-th cycle group of $S$ as
  $Z_k(S)=\ker \partial_{k}$. We have $B_k(S)\subset Z_k(S)$.
We are now able to define the $k$-th homology group and its dimension:
\begin{definition}
  The $k$-th homology group of an abstract simplicial complex $S$ is
  the quotient vector space: 
  \begin{eqnarray*}
    H_k(S)=\frac{Z_k(S)}{B_k(S)}.
  \end{eqnarray*}
 The $k$-th Betti number of the abstract simplicial complex $S$ is: 
  \begin{eqnarray*}
    \beta_k(S)=\dim H_k(S).
  \end{eqnarray*}
\end{definition}

According to its definition, the $k$-th Betti number counts the number
of cycles of $k$-simplices that are not boundaries of
$(k+1)$-simplices, that are the $k$-th dimensional holes. In small
dimensions, they have a geometrical interpretation:
\begin{itemize}
\item $\beta_0$ is the number of connected components,
\item $\beta_1$ is the number of coverage holes,
\item $\beta_2$ is the number of $3$D-voids.
\end{itemize}
For any $k\geq d$ where $d$ is the dimension, we have $\beta_k=0$.

For further reading on algebraic topology, see
\cite{hatcher_algebraic_2002}. 

\section{Algorithm}
\label{sec_alg}

Thanks to simplicial homology, we have a representation for a wireless
network that allows the computation of the network's topology, that is
its coverage, or $1$-coverage. Computing the $k$-coverage is another
problem. In order to do that, we choose to view the $k$-coverage as $k$ layers
of coverage:
\begin{lemma}
An area is $k$-covered, for $k$ integer, if there exists $k$ sets of
network nodes without any common nodes such that each set provides
$1$-coverage on the area.
\end{lemma}
\begin{proof}
If there exists $k$ sets of network nodes that provide $1$-coverage,
then let $x$ be any point in the area, $x$ is covered by each
layer. Thus, there exists $k$ nodes, one per layer, that cover
$x$. And the area is $k$-covered.

Reciprocally, it is not possible to find $k$ sets of nodes such that
each provide $1$-coverage. We can suppose that there exists $k-1$ sets
of nodes that provide exactly $1$-coverage, and a last set with the
remaining nodes that do not provide $1$-coverage. That is there exists
at least one coverage hole in the coverage provided by the $k$-th
set. Then let $x$ be a point in this coverage hole, then $x$ is in the
coverage range of exactly one node in each of  the first $k-1$ sets,
since these sets of node provide exactly $1$-coverage. The point $x$
is inside a coverage hole of the remaining nodes, then there is no
other node which coverage range covers $x$. And $x$ is not
$k$-covered.
\end{proof}

Therefore, to compute the $k$-coverage of a wireless network, we
intend to count the number of $1$-coverage layers. To slice the
network in layers, we use the simplicial complex reduction algorithm
that we presented in \cite{vergne_reduction_2013}. This reduction
algorithm takes as input a simplicial complex, then removes points and
their cofaces (that is the simplices they are part of) until it is no
more possible without creating neither a coverage hole nor a
disconnectivity in the network. At the end, we obtain a simplicial
complex that provides the same coverage as the initial complex with a
minimal set of points. Then the set of network nodes associated to
these points provide at least $1$-coverage on the whole covered area, but not
$2$-coverage on the whole area or more points could be removed and the
simplicial complex could be further reduced. However, locally
$2$-coverage is provided by the reduced complex, since we consider
coverage disks and disks can not tile the plane, there will always exist
intersection of disks. It is important to note that the reduction
algorithm needs the definition of a boundary (via a list of points) to
delimit the area to be covered, here it is the boundary of the area
where one need to compute the $k$-coverage. We can see an example of the
reduction algorithm on a Vietoris-Rips simplicial complex in
Fig.~\ref{fig_reduc}. 

\begin{figure}[h]
  \centering
    \includegraphics[width=4.35cm]{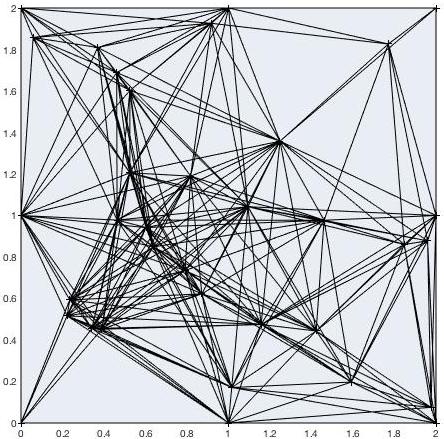}
 \includegraphics[width=4.35cm]{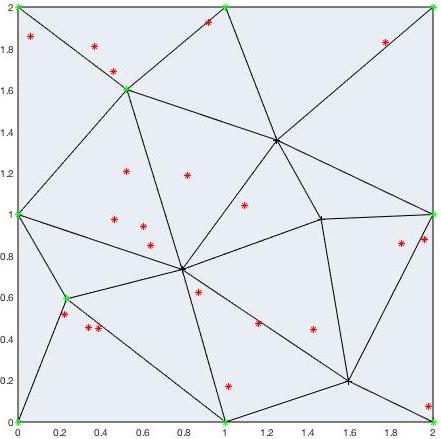}
  \caption{ A complex reduced by the reduction algorithm
    presented in \cite{vergne_reduction_2013}}
\label{fig_reduc}
\end{figure}

Our algorithm for computing $k$-coverage then takes as input the
positions of the network nodes, compute the simplicial complex to
represent their topology. Then, the reduction algorithm is applied,
its result constitutes of the first layer of $1$-coverage. This first
layer is then discarded, the simplicial complex is built on the
remaining points and the we re-apply the reduction algorithm on it. We
continue while the number of connected components stays at $1$, and
the number of coverage holes stays at $0$. At the end, our algorithm
provides the $k$ index of $k$-coverage of the wireless network, and
also supplies the $k$ sets of points/network nodes that are the $k$
layers of coverage. The pseudo-code of the algorithm is given in
Alg.~\ref{alg_couv}. 

\begin{algorithm}[H]
  \caption{$k$-coverage computing algorithm.}
\label{alg_couv}
  \begin{algorithmic}[h]
    \Require{set $V$ of $n$ vertices, coverage radius $r$.}
\State{$S:=\mathcal{R}_{2r}(V)$ or $\mathcal{C}_{r}(V)$\;}
\State{Computation of $\beta_0(S)$ and $\beta_1(S)$\;}
\State{$k:=0$\;}
\While{$\beta_{0}(S) =1$ and $\beta_{1}(S) =0$}
\State{$k=k+1$\;}
\State{Apply reduction algorithm to $S$\;}
\State{Save reduced complex as $k$-th layer\;}
\State{Save list of discarded vertices as $V'$\;}
\State{$S:=\mathcal{R}_{2r}(V')$ or $\mathcal{C}_{r}(V')$\;}
\State{Computation of $\beta_0(S)$ and $\beta_1(S)$\;}
\EndWhile{}
\Return{$k$ and the $k$ layers of coverage}
  \end{algorithmic}
\end{algorithm}

Our algorithm provides a lower-bound for $k$-coverage, that is that
$k$-coverage is guaranteed in every point of the area. More
specifically, when the algorithm returns the value $k$ for a wireless
networks on the area $A$, that means that:
\begin{itemize}
\item $\forall x \in A$, $x$ is $k$-covered,
\item $\exists x \in A$, $x$ is not $(k+1)$-covered,
\item There may exist some $x \in A$ that are $l$-covered with $l>k$ (at the
  intersection of coverage disks).
\end{itemize}

\section{Simulation results}
\label{sec_sim}
In this section, we give some figures illustrating the functioning of our
algorithm and present some simulation results on the
$k$-coverage of a wireless network simulated by a Poisson point
process.

We can see an example of the execution of our $k$-coverage computation
algorithm for the wireless network represented by a Vietoris-Rips
complex of Fig.~\ref{fig_rips}. This network was simulated with $N=40$
points randomly placed in a square of size $10$, plus a boundary of fixed
points on the square to delimit the area of the square in which we
want to compute the $k$-coverage. The coverage radius is set to $2.5$.
In order to obtain nicer figures, the process
used to draw the points positions is of hard-core type, that is it is
forbidden for $2$ points to be too close to each other.
\begin{figure}[h]
  \centering
    \includegraphics[width=4.35cm]{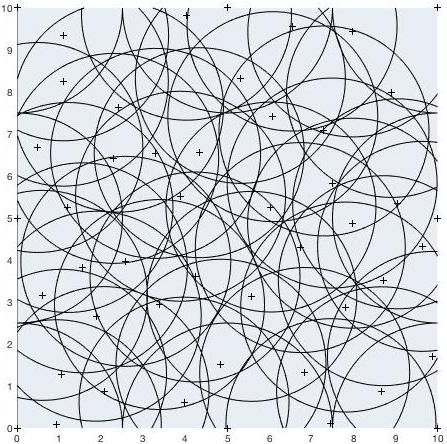}
 \includegraphics[width=4.35cm]{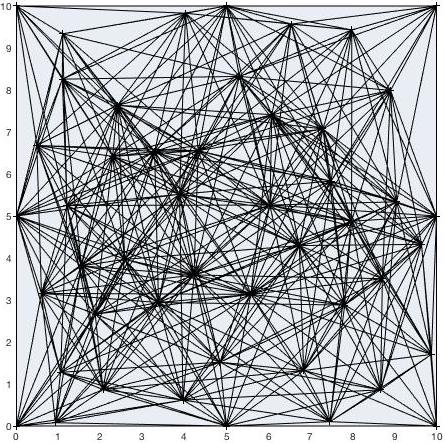}
  \caption{ A wireless network and its Vietoris-Rips representation}
\label{fig_rips}
\end{figure}

We can see in Fig.~\ref{fig_layers_rips} that our algorithm exhibits
$3$ layers of coverage, that means that the wireless network provides
$3$-coverage. The last layer presents a coverage hole in the bottom 
right corner, so $4$-coverage is not available in that part, and thus
on the square. In each subfigure, points in red are the remaining
points that are not yet part of a layer. On the left of each subfigure
is the wireless network representation with the coverage disks, and on
the right is the Vietoris-Rips representation.
\begin{figure}[h]
  \centering
 \includegraphics[width=4.35cm]{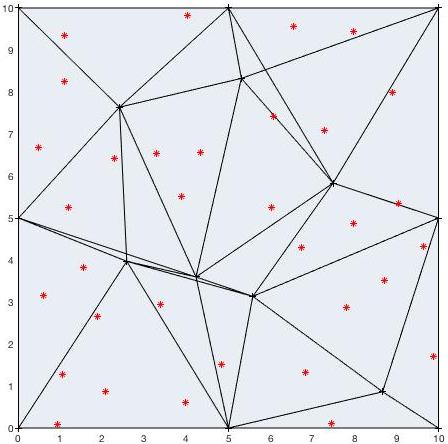}
 \includegraphics[width=4.35cm]{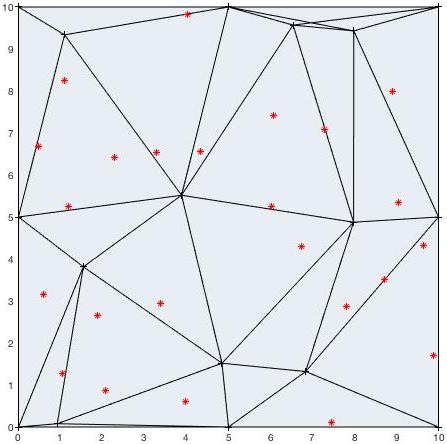}\\
 \includegraphics[width=4.35cm]{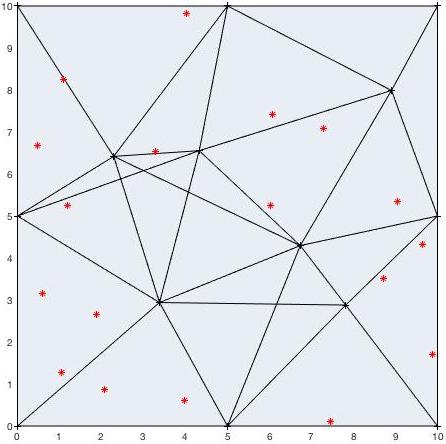}
 \includegraphics[width=4.35cm]{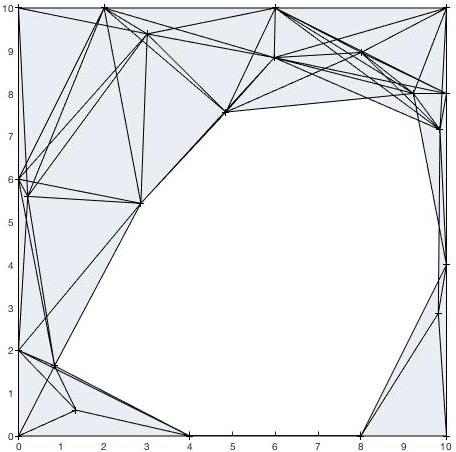}
  \caption{ The $3$ layers of coverage and the last incomplete layer.}
\label{fig_layers_rips}
\end{figure}

We also provide an example of our $k$-coverage algorithm running of the
\u{C}ech complex of Fig.~\ref{fig_cech}. The configuration set-up is the same as before,
except the number of points is initially set to $N=50$.
\begin{figure}[h]
  \centering
    \includegraphics[width=4.35cm]{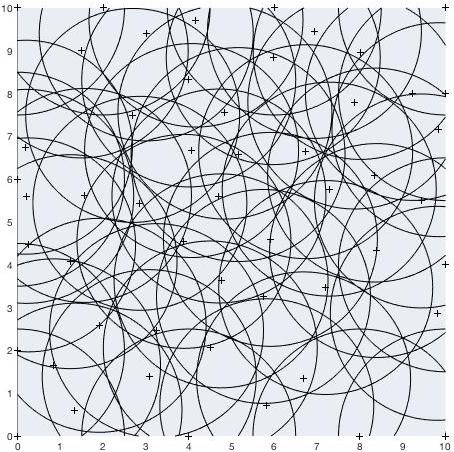}
 \includegraphics[width=4.35cm]{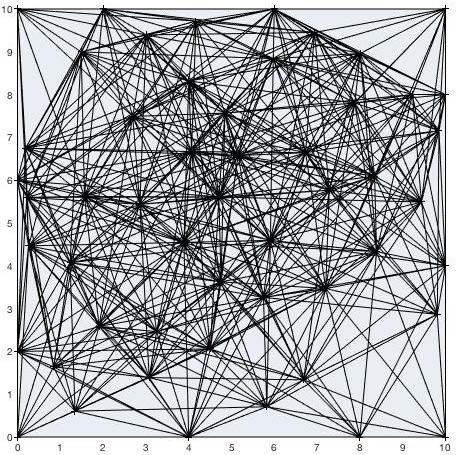}
  \caption{ A wireless network and its \u{C}ech representation}
\label{fig_cech}
\end{figure}

This network has more points and provides $4$ layers of coverage that
is $4$-coverage as we can see in Fig.~\ref{fig_layers_cech}.
\begin{figure}[h]
  \centering
 \includegraphics[width=4.35cm]{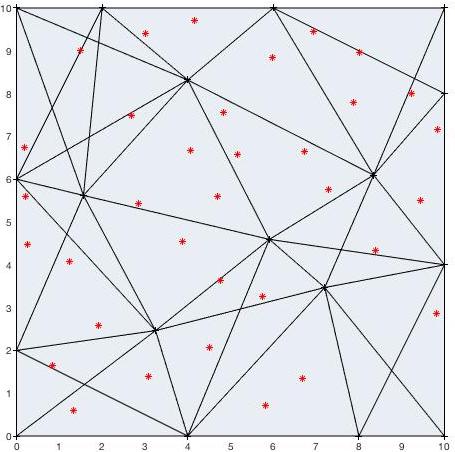}
 \includegraphics[width=4.35cm]{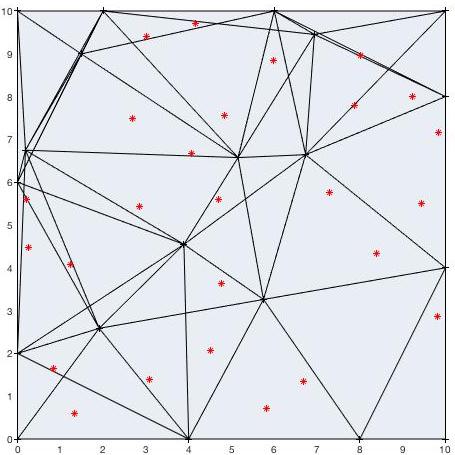}\\
 \includegraphics[width=4.35cm]{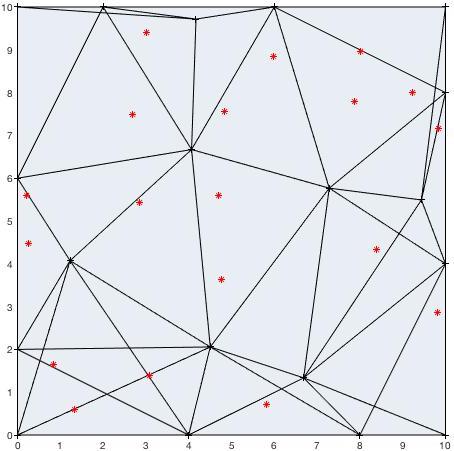}
 \includegraphics[width=4.35cm]{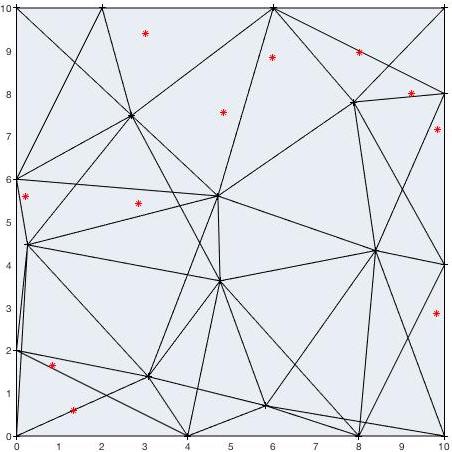}\\
 \includegraphics[width=4.35cm]{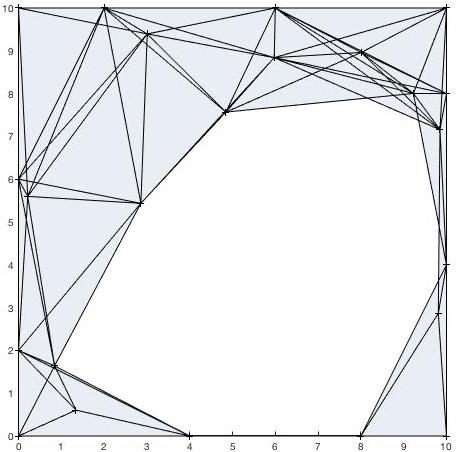}
  \caption{ The $4$ layers of coverage and the last incomplete layer.}
\label{fig_layers_cech}
\end{figure}

Finally, we provide some simulation results on the $k$-coverage of a
wireless network generated by a Poisson point process. We look at the
value of $k$, where $k$ is the maximum index such that the network
provides $k$-coverage, depending on the intensity of the process, that
is the mean number of points by surface unit, and whether the topology
is computed via a Vietoris-Rips or a \u{C}ech complex.

We consider a set of points generated with a Poisson point process of
intensity $\lambda$ on a square with side of size $10$. We add points
on the boundary of the square to delimit the area to be covered. The
coverage radius is set to $2.5$, one can note that $\pi r^2 \approx 20$.
We compute $\bar k$ the mean value of $k$ such that the complex
provides $k$-coverage and not $(k+1)$-coverage, on average on $1000$
simulations. 

\begin{table}[h]
\caption{Mean $\bar k$ for a Vietoris-Rips and for for a \u{C}ech complex}
\label{table}
\tiny
\begin{tabular}{|c|cccccccc|}
\hline
$\lambda$&$0.05$&$0.10$&$0.15$&$0.20$&$0.25$&$0.30$&$0.35$&$0.40$\\
\hline
$\lambda \pi r^2$&$1$&$2$&$3$&$4$&$5$&$6$&$7$&$8$\\
\hline
$\bar k$&$0.004$&$0.081$&$0.268$&$0.590$&$0.967$&$1.396$&$1.820$&$$\\
\hline
\end{tabular}

\smallskip

\begin{tabular}{|c|cccccccc|}
\hline
$\lambda$&$0.05$&$0.10$&$0.15$&$0.20$&$0.25$&$0.30$&$0.35$&$0.40$\\
\hline
$\lambda \pi r^2$&$1$&$2$&$3$&$4$&$5$&$6$&$7$&$8$\\
\hline
$\bar k$&$0.002$&$0.079$&$0.237$&$0.540$&$0.859$&$1.297$&$1.695$&$$\\
\hline
\end{tabular}
\end{table}
We can see in Table \ref{table} the result
of the simulations. Moreover, these results are plotted in the graph
of Fig.~\ref{fig_sim}. We can compare these results to the theoretic ones
of the mean $k$ for which a point $x$ is $k$-covered: $\E{k}=\lambda
\pi r^2$. We can see that, as expected, the simulated values
are below the theoretic ones. This is because, theoretically we are
only capable of computing the probability of a point to be
$k$-covered, but not the probability that a whole area is $k$-covered
without any coverage holes. The second one, that we can approach by
simulation, is smaller than the first one. Furthermore, our algorithm
guarantees $k$-coverage, but locally points may be $l$-covered with
$l>k$.
\begin{figure}[h]
  \centering
    \includegraphics[width=7cm]{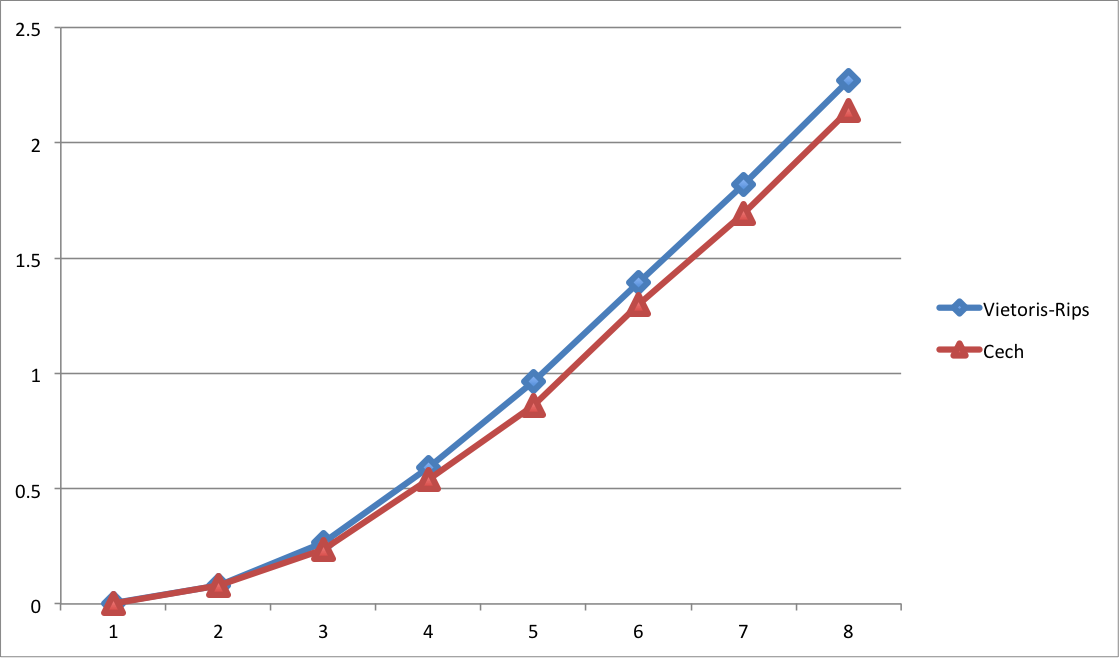}
  \caption{Mean $\bar k$ computed by our algorithm.}
\label{fig_sim}
\end{figure}

\section{Conclusion}
\label{sec_ccl}
In this article, we propose a method and an algorithm for computing
the $k$-coverage of a wireless network. The $k$-coverage is the fact
for a point to be covered by $k$ network nodes, this definition can be
extended to a whole area: an area is $k$-covered if every point in it
is $k$-covered. Theoretically it is easy to compute the probability
for a point to be $k$-covered for a wireless network generated by a
Poisson point process. However it is more difficult to apprehend the
$k$-coverage for a whole area, and we need simplicial homology
representation to compute the topology of the network as a whole. We
then exhibit that the $k$-coverage can be seen as $k$ layers of
$1$-coverage and give an algorithm that compute the $k$-coverage of a
wireless network. We provide some figures and simulation results
to illustrate our algorithm.


\end{document}